\documentclass[a4paper,USenglish]{lipics-v2021}
\pdfoutput=1
\nolinenumbers
\hideLIPIcs
\usepackage{amsthm,amsmath,bm,mathtools,makecell,threeparttable,tabularx,aliascnt,pifont,booktabs}

\usepackage[noend]{algpseudocode}
\usepackage[many]{tcolorbox}
\newenvironment{algobox}[2][]{
    \begin{center}\begin{tcolorbox}[enhanced,title=\centering \large {#2},colback=white,colframe=black!80,width=\columnwidth,breakable]
}{
    \end{tcolorbox}\end{center}
}

\makeatletter \newcommand{\apxthm}[3]{
\newcounter{sec#2} \newcounter{val#2} \newcounter{use#2}
\setcounter{sec#2}{\value{section}} \setcounter{val#2}{\value{theorem}}
\begin{#1} \label{#2} {#3} \end{#1}
\long\@namedef{#2}{
\gdef\thesection{\@arabic\c@section}
\@tempcnta=\value{section} \@tempcntb=\value{theorem}
\setcounter{section}{\value{sec#2}} \setcounter{theorem}{\value{val#2}}
\phantomsection \label{#2*} \begin{#1} {#3} \end{#1} 
\setcounter{section}{\@tempcnta} \setcounter{theorem}{\@tempcntb}
\gdef\thesection{\@Alph\c@section}
}}\makeatother

\algblock{Upon}{EndUpon} \algrenewtext{Upon}[1]{\textbf{upon} {#1} \textbf{do}}
\makeatletter \ifthenelse{\equal{\ALG@noend}{t}} {\algtext*{EndUpon}} \makeatother

\DeclarePairedDelimiter{\floor}{\lfloor}{\rfloor}
\DeclarePairedDelimiter{\ceil}{\lceil}{\rceil}

\newcommand{\llabel}[1]{\phantomsection\label{line:#1}}

\newcommand{\Pub}{\mathit{Publish}}
\newcommand{\C}{\mathrm{C}}
\newcommand{\str}{\textsf{str}}

\newaliascnt{fullthm}{theorem}

\aliascntresetthe{fullthm}

\title{Subcubic Coin Tossing in Asynchrony without PKI}

\author{Mose Mizrahi}{ETH Zurich, Switzerland}{mmizrahi@ethz.ch}{https://orcid.org/0009-0009-9771-0845}{}
\author{Roger Wattenhofer}{ETH Zurich, Switzerland}{wattenhofer@ethz.ch}{https://orcid.org/0000-0002-6339-3134}{}
\authorrunning{M.\ Mizrahi and R.\ Wattenhofer}
\Copyright{Mose Mizrahi Erbes, and Roger Wattenhofer}
\acknowledgements{We thank Ittai Abraham and Gilad Asharov for assisting us regarding the AVSS protocol in \cite{aaps24}.}

\begin{document}

\maketitle

\begin{abstract}
We consider an asynchronous network of $n$ parties connected to each other via secure channels, up to $t$ of which are byzantine. We study common coin tossing, a task where the parties try to agree on an unpredictable random value, with some chance of failure due to the byzantine parties' influence. Coin tossing is a well-known and often-studied task due to its use in byzantine agreement.

In this work, we present a committee-based method to transform strong (rarely failing) binary common coins into weaker ones that asymptotically require less communication. For any $k > 2$ and $\varepsilon > 0$, we can transform a strong binary coin that costs $\widetilde{O}(n^k)$ bits of communication into a weak binary coin that costs $\widetilde{O}(\varepsilon^{-2k}n^{3 - 2/k})$ bits. This latter coin tolerates $\varepsilon n$ fewer byzantine parties than the strong coin it is based on, and it fails with an arbitrarily small constant probability.

With our method, we obtain a secure-channel-based perfectly secure coin for $t \leq (\frac{1}{4} - \varepsilon)n$ faults that costs $\widetilde{O}(n^{2.5})$ bits, as well as a coin based on cryptographic hashing for $t \leq (\frac{1}{3} - \varepsilon)n$ faults that costs $\widetilde{O}(n^{7/3}\kappa)$ bits. These are to our knowledge the first PKI-free asynchronous common coins that cost $o(n^3)$ bits of communication but still succeed with at least constant probability against $t = \Theta(n)$ \linebreak adaptive byzantine faults.
\end{abstract}

\section{Introduction}

In byzantine agreement (BA) \cite{psl80, lsp82}, there is a network of $n$ parties with inputs who want to agree on an output. If the parties have a common input, then they must agree on it. The challenge is that an adversary corrupts up to $t$ of the $n$ parties, rendering them byzantine, and the byzantine parties arbitrarily deviate from the BA protocol that the parties run.

BA is particularly challenging in asynchronous networks, where the adversary can delay any message as much as it wants. By the FLP theorem \cite{flp85}, deterministic BA is impossible in an asynchronous network even if just a single party can crash. To circumvent this impossibility result, asynchronous BA protocols use randomization. The foremost randomization technique, pioneered by Ben-Or \cite{benor83} and Rabin \cite{rab83}, is to use a \emph{common coin}. A common coin protocol allows the parties to agree on a random value which the adversary cannot predict in advance. We define common coins as follows:

\begin{definition}[Common Coin]
    A common coin $\C$ is a protocol which a set of parties without inputs run to get outputs in an output domain $D$. If $\C$ is $\delta$-fair, then it achieves the following: \begin{itemize}
        \item \textbf{liveness:} If the honest (never-corrupted) parties all run a common instance of $\C$, then they all obtain outputs in $D$ from the instance.
        \item \textbf{$\bm\delta$-fairness:} Each $\C$ instance is associated with two independent random variables sampled freshly for the instance: a uniformly random $v \in D$ and a bit $g$ such that $\Pr[g = 1] = \delta$. If $g = 1$, then the honest parties can only output $v$ from the instance, and the adversary does not know anything about $(v,g)$ until some honest party becomes active in the instance. \end{itemize}
\end{definition}
Note that one can further formalize this definition with an ideal common coin functionality, like \cite{cfgpz23} does. In some works (e.g.\ \cite{cfgpz23}), a common coin as we have defined is called \emph{oblivious} because the honest parties do not know if they have agreed on a random value or not.

Ben-Or's common coin is binary (produces outputs in $\{0,1\}$). It tolerates $t < \frac{n}{3}$ byzantine faults, the maximum number of faults any asynchronous BA protocol can tolerate \cite{t84,bracha84}. It is simple. It works as follows: Each party generates a random bit and broadcasts it. Then, each party waits to receive from $n - t$ parties, and after doing so it outputs the majority bit it received (breaking ties arbitrarily). This coin works because for any bit $b$, if every honest party randomly generates $b$, then $b$ ends up being the majority bit for everyone, and so every honest party outputs $b$. Moreover, before some honest party generates a bit, the adversary cannot with 100\% accuracy predict what the common output $b$ (if any) will be.

The drawback of Ben-Or's coin is its low fairness. For each bit $b$, the coin only guarantees the common output $b$ if at least $\ceil{\frac{n + t + 1}{2}}$ honest parties randomly generate $b$. When $t = \Theta(n)$, the probability of this is $2^{-\Theta(n)}$, which means that the coin is only $2^{-\Theta(n)}$-fair. Consequently, BA protocols that use this coin \cite{benor83,bracha84} have exponential-in-$n$ expected latencies. Note however that when $t = O(\sqrt{n})$, the probability that enough honest parties generate either bit becomes $\Omega(1)$, and the coin thus achieves constant fairness ($\delta$-fairness for a constant $\delta > 0$) \cite{benor83}.

The threshold $t = \widetilde{\Theta}(\sqrt{n})$ is known to be critical for \emph{adaptive} adversaries, which can corrupt parties during a protocol's execution based on the information they gather while the parties run the protocol. This is because adaptive adversaries that can corrupt $t = {\omega}(\sqrt{n} \cdot \mathrm{polylog}(n))$ parties make it impossible to toss common coins with constant fairness \cite{hk26}. Nevertheless, there exist common coins which scale better with $n$ when $t = \Theta(n)$ than Ben-Or's. In \cite{hsl24}, Huang, Pettie and Zhu present an $\Omega(1/\sqrt{n})$-fair asynchronous binary coin based on statistical fraud detection that can tolerate $t < \frac{n}{3 + \varepsilon}$ faults with $\widetilde{O}(n^{4}/\varepsilon^8)$ latency for any $\varepsilon > 0$.

The critical threshold $t = \widetilde{\Theta}(\sqrt{n})$ is critical only in the \emph{full-information} setting, where the adversary is computationally unbounded and omniscient. This makes the adversary very strong. For example, in a Ben-Or coin execution where a slim majority of the honest parties randomly generate some bit $b$, the full-information adversary can counteract this majority by making the byzantine parties broadcast $1-b$. The honest parties cannot prevent this by somehow obscuring their bits because the adversary knows everything.

Due to the difficulty of adaptive fault tolerance in the full-information setting, common coins often target more restricted adversaries. For example, one can target static adversaries that must choose the byzantine parties before protocol execution. Kapron et al.\ show in \cite{kapron10} that for every absolute constant $\varepsilon > 0$ there is a constant-fairness leader election protocol that tolerates up to $t \leq (\frac{1}{3} - \varepsilon)n$ static byzantine faults, with $\widetilde{O}(1)$ latency and $\widetilde{O}(n^2)$ bits of communication \cite{kapron10}. Note that a leader election protocol is a common coin which has the output domain $[n] = \{1,\dots,n\}$, where the coin's output determines a leader party in $[n]$.

In this work, we instead limit what the adversary can know. We assume that the parties are connected to each other via secure channels that ensure privacy (and that the adversary cannot see the honest parties' internal states). This means that the honest parties can share random secrets with each other, and then reveal the secrets when the adversary is too late to stage secret-dependent attacks. This popular strategy was pioneered by Canetti and Rabin, who used secure channels to get an efficient $\frac{1}{2}$-fair binary coin that tolerates $t < \frac{n}{3}$ faults \cite{cr93}.

We need to assume secure channels because they enable \emph{strong} common coins, which can achieve $\delta$-fairness for any $\delta \in (0, 1)$ with reasonable efficiency even when $1 - \delta$ is negligibly small. Informally, our work is about a method to transform any strong binary coin that costs $\widetilde{O}(n^k)$ bits of communication for some $k > 2$ into a weaker (less fair) but also cheaper binary coin that costs $\widetilde{O}(\varepsilon^{-2k}n^{3 - 2/k})$ bits of communication, while lowering the fault tolerance $\frac{t}{n}$ by only $\varepsilon > 0$. Our method allows us to obtain the first common coins (and thus BA protocols, by \cite{mmr15}) that tolerate $t = \Theta(n)$ adaptive byzantine faults with $o(n^3)$ bits of communication without PKI (public key infrastructure) setups. One of our coins will even be perfectly secure, solely reliant on a secure channel setup and no other assumptions.

\section{Contributions} \label{sec-contributions}

Our main contribution is a technique to decrease the asymptotic message/communication complexity of common coin protocols that cost $n^{2 + \Omega(1)}$ messages/bits. We do so by designing a transformation which converts strong but expensive binary common coins into weaker but also cheaper ones. Below, Theorem \ref{simplethm} states the relation between the strong coin $\C_\str$ the transformation takes in and the weaker coin $\mathit{T}(\C_\str,z,k,\varepsilon,\alpha)$ we thus obtain.

\begin{theorem}[Simplified]
\label{simplethm}
    For any positive real constants $z < 1$, $\alpha < \frac{1}{3}$, $\delta \leq 1$, $k \geq 2$ and any real parameter $0 < \varepsilon < \alpha$, suppose we have a $\delta$-fair binary common coin $\C_\str$ such that when $n$ parties run it, $\C_\str$ can tolerate up to $t < \alpha n$ adaptive byzantine faults, and it costs $\widetilde{O}(n^k)$ \linebreak messages of size at most $L(n)$ for a non-decreasing function $L$. Then, we can obtain a binary common coin $\mathit{T}(\C_\str,z,k,\varepsilon,\alpha)$ with the following properties: \begin{itemize}
        \item $\mathit{T}(\C_\str,z,k,\varepsilon,\alpha)$ tolerates up to $t \leq (\alpha - \varepsilon)n$ adaptive byzantine faults.
        \item $\mathit{T}(\C_\str,z,k,\varepsilon,\alpha)$ is $(\delta^q(1 - z))$-fair, where $q = 2\ceil{n^{2-2/k}} + 1$.
        \item $\mathit{T}(\C_\str,z,k,\varepsilon,\alpha)$ costs $\widetilde{O}(\varepsilon^{-2k}n^{3-2/k})$ messages of size at most $L(n) + O(\log n)$.
        \item If $\C_\str$'s latency is at most $R$ when $n$ or less parties run it, then $\mathit{T}(\C_\str,z,k,\varepsilon,\alpha)$'s latency is at most $R + 5$ when $n$ parties run it.
    \end{itemize}
\end{theorem}

To keep the exposition light, we only state a simple version of Theorem \ref{simplethm} here. In the full version in Section \ref{sectrans}, we treat $z,\alpha,\delta,k$ as variable parameters, we drop the assumption that $\C_\str$'s message complexity is $\widetilde{O}(n^k)$, we let $L(n)$ be an arbitrary function, and we state $\mathit{T}(\C_\str,z,k,\varepsilon,\alpha)$'s complexity more exactly.

Note that the latency bounds in Theorem \ref{simplethm} only hold when the a tolerable fraction of the parties are corrupted: less than $\alpha$ for $\C_\str$, at most $\alpha - \varepsilon$ for $\mathit{T}(\C_\str,z,k,\varepsilon,\alpha)$. However, when we mention a message/communication complexity bound in the theorem, we mean that this bound holds no matter how many parties are corrupted, for both $\C_\str$ and $\mathit{T}(\C_\str,z,k,\varepsilon,\alpha)$.

Theorem \ref{simplethm} holds when the parties are connected via reliable and authenticated channels, even in the full-information setting. However, $\C_\str$ may require more security assumptions (such as secure channels), and if so, then $\mathit{T}(\C_\str,z,k,\varepsilon,\alpha)$ inherits these assumptions.

\subparagraph{Our method.} We prove Theorem \ref{simplethm} with a committee-based transformation. Suppose for example that we have a strong binary coin $\C_\str$ that tolerates $t < n/4$ faults and costs $O(n^k)$ words (messages of size $O(\log n)$) for some $k > 2$, and we want to get a cheaper binary coin that tolerates $t < n/5$ faults. Then, we assort the $n$ parties into $q = \Theta(n^{2-2/k})$ deterministically fixed committees (party subsets) of size $s = \Theta(n^{1/k} + \log n)$, with the safety guarantee that no matter which $t < n/5$ parties the adversary adaptively corrupts, only $O(\sqrt{q})$ committees can become bad (start containing $\frac{s}{4}$ or more corrupted parties). Each committee internally generates a random bit by running $\C_\str$, and then publishes its bit $b$ to make $b$ known to the parties outside. Finally, each party outputs the bit that it thinks a majority of the committees published. Essentially, we mimic a $q$-party instance of Ben-Or's common coin by making each committee act as a virtual party, a virtual instance with at least constant fairness as only \linebreak $O(\sqrt{q})$ of the virtual parties (committees) are corrupt (bad). 

The steps above cost $O(qs^k + qs^2 + qn\log n + n^2) = O(n^{3 - 2/k}\log n)$ words in total. The $qs^k$ term arises from each committee internally generating a bit with $O(s^k)$ words by running $\C_\str$, whereas the $qs^2$ and $qn\log n$ terms arise from each committee publishing its bit with a protocol where the committee first reaches crusader agreement\footnote{Crusader agreement \cite{d82, crusader-modern, aw24} is a relaxation of BA. We use it to prevent our construction from losing liveness when the honest parties in a good (non-bad) committee do not agree on which bit they generated. Such a disagreement might occur due to the committee not using a 1-fair strong coin.} on its bit with $O(s^2)$ words and then sends $O(n\log n)$ words in total to the parties outside the committee. Finally, the $n^2$ term is due to a final step where the honest parties who have seen enough committees' bits each broadcast their majority bits so that everyone can output.

While we were not aware of this when we thought of it, our core idea of using committees as virtual parties in a virtual execution of Ben-Or's coin dates back to Bracha's work from the 1980s \cite{bracha87}. Bracha used $\Theta(n^2)$ committees of size $\Theta(\log n)$ to solve randomized synchronous byzantine broadcast in $O(\log n)$ expected rounds, with $\Theta(n^4\log^2 n)$ messages per round.

Our insight that allows us to improve on Bracha's work is that his virtual party simulation is too faithful. In \cite{bracha87}, the virtual parties (committees) run Ben-Or's coin as follows: \begin{enumerate}
    \item Each virtual party internally generates a random bit via a strong common coin protocol.
    \item For every pair of virtual parties $Q$ and $Q'$, the virtual party $Q$ sends $Q'$ the bit $b$ that $Q$ generated. This involves each real party in $Q$ sending $b$ to each real party in $Q'$, and then \linebreak the real parties in $Q'$ internally running a BA protocol to agree on $b$.
\end{enumerate}

In the second step, the virtual parties simulate all the message sending/receiving operations of Ben-Or's coin, and thus simulate a Ben-Or coin execution. In this simulation, each member of a committee $Q$ sends $Q$'s bit to each real party $i$ once for every committee that contains $i$. We remove this unnecessary message duplication. Then, we observe that BA is not needed: It is fine if each real party (rather than each committee) forms its individual opinion on the published bits. With the message duplication and the BA instances removed, the cost of $q$ committees of size $s$ publishing their bits drops from $O(q^2(s^2 + A(s)))$ words (where $A(s)$ is the cost of each BA instance) to $O(qsn)$. To further reduce $O(qsn)$ to $O(qs^2 + qn\log n)$, we observe that the committees need not publish their bits perfectly: Some miscommunication can be tolerated as long as $\floor{\frac{2n}{3}} + 1$ real honest parties correctly learn all but $O(\sqrt{q})$ of the good committees' bits. Therefore, we can design and use an imperfect bit publishing protocol where each committee $Q$ reaches crusader agreement on the bit it generated with $O(s^2)$ words, and then the members of $Q$ send their crusader agreement outputs to $O(\frac{n \log n}{s})$ parties outside $Q$ on average per committee member. Finally, we choose $q$ and $s$ in a way which balances the complexity of strong coin tossing and bit publishing, while respecting the safety constraints \linebreak $qs^2 = \Omega(n^2)$ and $s = \Omega(\log q)$ so that only $O(\sqrt{q})$ of the committees can be bad.

Roughly speaking, given a strong coin costing $\widetilde{O}(n^k)$ bits, we let $s \approx n^{1/k}$ and $q \approx n^{2 - 2/k}$ to balance the strong coin tossing cost $\widetilde{O}(qs^k)$ and the bit publishing cost $O(qn\log n)$. This leads to a total cost of $\widetilde{O}(n^{3-2/k})$ bits. We note that even without the imperfect bit publishing (i.e. with the bit publishing cost $O(qsn)$), we could get a subcubic total cost of $\widetilde{O}(n^{3 - 1/(k-1)})$ by setting $s \approx n^{1/(k-1)}$ and $q \approx n^{2 - 2/(k-1)}$; however, this would be worse than $\widetilde{O}(n^{3-2/k})$.

We get our bit publishing protocol by observing the following fact: Given a committee $Q$ of size $s = O(\sqrt{n})$ which contains at most $\frac{n}{3}$ corrupt parties, we can deterministically fix $n$ subsets $Q_1, \dots, Q_n \subseteq Q$ of size $O(\log n)$ such that all but $O(\sqrt{n})$ of these subsets are majority-honest. This means that we can task each subset $Q_i$ with sending the committee's bit to the party $i$, and doing so will result in a large set $P^*$ of $n - O(\sqrt{n})$ parties such that each party $i \in P^*$ can learn the committee's bit correctly by receiving it from the majority-honest subset $Q_i$. We will formalize this idea with sparse bipartite graphs. We note that one can find similar uses of sparse bipartite graphs in the prior committee-based agreement literature. \linebreak In particular, \cite{ks10} uses them to solve byzantine agreement with $\widetilde{O}(n\sqrt{n})$ bits of communication in a synchronous network where an adaptive adversary corrupts $t \leq (\frac{1}{3} - \varepsilon)n$ parties. Note that when viewed as functions mapping each party $i \in [n]$ to the subset $Q_i \subseteq Q$, the bipartite graphs we use in this work are closely related to averaging samplers \cite{ks10}.

\begin{table}[ht]
\centering 
\begin{threeparttable}
    \renewcommand{\tnote}[1]{\textsuperscript{#1}}
    \renewcommand{\arraystretch}{1.15}
    \setlength{\aboverulesep}{0.4pt}
    \setlength{\belowrulesep}{0.4pt}
    \setlength{\tabcolsep}{4.32pt}
    \caption{The most efficient protocols we know for asynchronous coin tossing with at least constant fairness against $t = \Theta(n)$ strongly adaptive byzantine faults, assuming secure channels. With $\kappa$ we represent a security parameter, and we assume that $\kappa = \Omega(\log n)$.}\label{compare}
    \label{comparisons}
    \begin{tabular}{cccccc} 
        \toprule
        \textbf{Threshold} & \textbf{Security} & \textbf{PKI-Free?} & \textbf{Communication (Bits)} & \textbf{Latency} & \textbf{Source} \\ \midrule
        $t < \frac{n}{4}$ & perfect & \ding{51} & $O(n^4\log n)$ & $O(1)$ & \cite{aaps24} \\
        $t < \frac{n}{3}$ & crypto & \ding{51} & $O(n^3\kappa\log n)$ & $O(\log n)$ & \cite{fkt22}\tnote{a} \\
        $t < \frac{n}{3}$ & crypto & \ding{51} & $O(n^3\kappa)$ & $O(1)$ & \cite{ddlmrs24}\tnote{b} \\
        $t < \frac{n}{3}$ & crypto & \ding{55} \ (plain PKI) & $O(n^2\kappa^2)$ & $O(\log n)$ & \cite{abls25}\tnote{c} \\
        $t < \frac{n}{3}$ & crypto & \ding{55} \ (silent setup)\tnote{d} & $O(n^2\kappa)$ & $O(1)$ & \cite{ft25} \\
        $t < \frac{n}{2}$ & crypto & \ding{55} \ (unique t.\ sigs)\tnote{e} & $O(n^2\kappa)$ & 1 & \cite{cks00} \\ \midrule
        $t \leq (\frac{1}{4} - \varepsilon)n$ & perfect & \ding{51} & $O(n^{2.5}(\varepsilon^{-8} + \log n)\log n)$ & $O(\log n)$ & this paper \\
        $t \leq (\frac{1}{3} - \varepsilon)n$ & crypto & \ding{51} & $O(n^{7/3}\varepsilon^{-6}\kappa\log n)$ & $O(\log n)$ & this paper \\
        \bottomrule
    \end{tabular}
    \begin{tablenotes}
        \item[a] In \cite{fkt22} there is also a common coin without the $\log n$ factor in the communication complexity and latency, but this coin is not proven to be adaptively secure.
        \item[b] The hash-based protocol in \cite{ddlmrs24} should be adaptively secure according to one of its authors \cite{shoup24}.
        \item[c] The common coin in \cite{abls25} can be configured to achieve the latency $O(r)$ for any $r \in \{1,\dots,\log n\}$, with this costing $O(n^{2 + r}\kappa^2)$ bits. The row in the table is for $r = O(\log n)$.
        \item[d] In \cite{ft25}, the authors rely on threshold signatures and a common reference string (CRS). Informally, the setup is silent since the parties can generate their keys without interacting with each other.
        \item[e] While the common coin in \cite{cks00} is more efficient than the one in \cite{ft25}, the coin in \cite{cks00} requires the stronger setup assumption of unique threshold signatures.
    \end{tablenotes}
\end{threeparttable} 
\end{table}

In Table \ref{compare}, we compare some adaptively secure common coins from the literature with what our method can produce. The table excludes some (PKI-based) coins that can tolerate $t = \Theta(n)$ faults with $\widetilde{O}(n \cdot \mathrm{poly}(\kappa))$ bits of communication \cite{bkll20,acgsy25,cks20}. The reason is that these coins do not achieve \emph{strongly} adaptive security; that is, adaptive security even when the adversary can drop messages by adaptively corrupting their senders. While we are not aware of any result showing that common coins with subquadratic message complexities cannot tolerate $t = \Theta(n)$ strongly adaptive faults, we note that such a lower bound exists for BA \cite{acdnprs19}.

There are many works in the literature that use committees for common coin tossing and for related tasks such as distributed key generation; see for example \cite{bracha87,ks10,kapron10,flt24,acgsy25,abls25,bkll20,cks20}. While random committees are more commonly used, there are works that use deterministic committees too, like the aforementioned work of Bracha \cite{bracha87} as well as a recent paper on subcubic-complexity synchronous distributed key generation \cite{flt24} which shards the $n$ parties into $\sqrt{n}$ committees of size $\sqrt{n}$, at least one of which is majority-honest due to the adversary corrupting $t < \frac{n}{2}$ parties. Adaptive adversaries can complicate the use of random committees because if one uses random committees naively, then one will be vulnerable to the fact that the adaptive adversary can adaptively corrupt committee members after they are randomly chosen. Still, one can combat this vulnerability and use random committees against adaptive adversaries \cite{bkll20,cks20,abls25,acgsy25}. In particular, \cite{abls25} even achieves strongly adaptive security. However, \cite{bkll20,cks20,abls25,acgsy25} all require PKI setups, for tasks such as message signing and committee sampling. In this work, we want to see what we can achieve without PKI and even without cryptography, even if this causes our common coins to not be as efficient as PKI-based ones.

Our method can produce a $\delta$-fair coin for any fairness parameter $\delta \in (0, 1)$. However, if $\delta = 1 - o(1)$, then it requires a stricter bound than $O(\sqrt{q})$ on the number of bad committees. For this, we let the committee size $s$ scale with $\frac{1}{1 - \delta}$ until when it reaches $n$. This means that \linebreak our method should only be used when $\delta = 1 - \Omega(\frac{1}{n})$. When $\delta$ is closer to $1$, one gets $s = n$, and committees of size $n$ do not make sense to use.

In Section \ref{strong-coin-chapter}, we use our method. We get the strong coins required for this via the Monte Carlo framework of \cite{fkt22} that allows one to get strong coins via asynchronous verifiable secret sharing (AVSS \cite{cr93}). With this framework, we obtain two strong coins: A perfectly secure one costing $O(n^4)$ words, and a cryptographically secure one costing $O(n^3\log n)$ messages of size $O(\kappa)$ (where $\kappa = \Omega(\log n)$ is a security parameter). Then, by applying our transformation, we get a perfectly secure binary coin $\C_\mathsf{perfect}$ and a cryptographically secure one $\C_\mathsf{crypto}$, which can both be set to achieve $\delta$-fairness for any $\delta \in (0, 1)$. Below, we assume $\delta = 1 - \Omega(\frac{1}{n})$. \begin{itemize}
    \item The coin $\C_\mathsf{perfect}$ is perfectly secure, which means that it achieves its liveness and $\delta$-fairness guarantees without any conditions. For any $\varepsilon > 0$, it can be set to tolerate $t \leq (\frac{1}{4} - \varepsilon)n$ faults. This coin costs $O(n^{2.5}(\varepsilon^{-8}(1-\delta)^4 + \log n))$ words, and its latency is $O(\log n)$. To obtain it, we use the perfectly secure AVSS scheme from \cite{aaps24}.
    \item The coin $\C_\mathsf{crypto}$ is cryptographically secure. For any $\varepsilon > 0$, $\C_\mathsf{crypto}$ can be set to tolerate $t \leq (\frac{1}{3} - \varepsilon)n$ faults. It costs $O(n^{7/3}(1-\delta)^3\varepsilon^{-6}\log n)$ messages of size at most $O(\kappa)$, and its latency is $O(\log n)$. To obtain it, we use the hash-based AVSS scheme from \cite{hashrand}, which is \linebreak secure in the programmable random oracle (pRO) model.\footnote{While the authors of \cite{hashrand} prove their AVSS scheme's security in the non-standard pRO model, they state that with recent techniques \cite{shoup24,ss24} the pRO model can be replaced with the more concrete assumptions about the hash function that it is collision-resistant and linear-hiding.}
\end{itemize}

Note that even though secure channels are usually implemented via cryptography, there are good reasons to distinguish secure channel use from cryptography use. There are ways to get secure channels without cryptography (using channel noise \cite{wyner} or quantum communication \cite{quantum-channels}), and it has recently been shown that one can transform classical BA protocols that use secure channels into full-information quantum ones \cite{quantumba}.

One can obtain a binary asynchronous BA protocol with $O(1)$ expected sequential tosses of a constant-fairness binary coin, with an expected $O(1)$ latency and $O(n^2)$ constant-size message complexity overhead \cite{mmr15}. Hence, the coins in Table \ref{comparisons} also lead to BA protocols with the same communication complexities and latencies shown in the table (in expectation). To our knowledge, our common coins are the first PKI-free ones in the literature that tolerate $t = \Theta(n)$ adaptive faults with $o(n^3)$ bits of communication (and even $o(n^4)$ bits for perfect security), and they lead to the first such asynchronous BA protocols.

While our method is primarily intended for binary coin tossing, the coins $\C_\mathsf{perfect}$ or $\C_\mathsf{crypto}$ can be used for leader election too, with $\widetilde{O}(1)$ times as much communication as using them to obtain binary coin tosses with the same fairness. This is because one can for any $\delta \in (0, 1)$ and $\ell \geq 1$ get a $\delta$-fair $\ell$-bit coin toss by tossing a $(1 - \frac{1 - \delta}{\ell})$-fair binary coin $\ell$ times in parallel to determine each bit of the $\ell$-bit output. We discuss this in Section \ref{strong-coin-chapter}.

Both $\C_\mathsf{perfect}$ and $\C_\mathsf{crypto}$ have the latency $O(\log n)$. The reason is that the strong coins we use to obtain them take $\Theta(r)$ rounds to achieve $(1-2^{-r})$-fairness, and we need $r = \Theta(\log n)$ for them to be fair enough. To achieve the latency $O(1)$ with out method, we would need to use a constant-latency strong coin. For instance, if we used the one-round strong coin from \cite{cks00} which is based on unique threshold signatures, then the latency of the transformed coin would be $6$. However, this transformed coin would be less efficient than the original strong coin, and it would also inherit the unique threshold signature setup assumption. As far as we know, there are no constant-latency alternatives to the $\Theta(\log n)$-latency strong coins we use in this work that would suit our purposes. So, it remains open for future work to bring the latency further down to $O(1)$ while keeping the communication subcubic.

Our main result is Theorem \ref{simplethm}. With this theorem, we break the $\Omega(n^3)$ communication complexity barrier for asynchronous coin tossing and thus for BA against $t = \Theta(n)$ adaptive byzantine faults without PKI, and even with perfect security using only secure channels. To our knowledge, it was not known before that BA against $t = \Theta(n)$ adaptive faults can be solved with $o(n^3)$ communication without at least using a PKI setup, and even the fact that one does have to use unique threshold signatures \cite{cks00} for this was only discovered in the last two years \cite{ft25,abls25,acgsy25}. Our transformation is fundamentally a refinement of Bracha's \cite{bracha87}, and we find it surprising that our adjustments (which do not rely on any complex cryptographic \linebreak machinery) remained undiscovered for so long. We hope that this paper will attract more attention to what one can do without PKI, and even without any cryptography at all. In particular, it remains open to bring the asymptotic communication complexity of adaptively secure asynchronous common coin tossing against $t = \Theta(n)$ faults without PKI further down to $n^{2 + o(1)}$, and we find it plausible that the core ideas required for this are already present in the literature. Practical improvements would also be welcome, given that our transformation is primarily of theoretical interest due to the large overheads involved for realistic values of $n$.

\section{Model}

We consider an asynchronous network of $n$ parties $1, \dots, n$, fully connected via reliable and authenticated channels. Each party knows the other parties' IDs.

We design our protocols against an adaptive adversary that corrupts up to $t$ of the parties during protocol execution, depending on the information it gathers while the parties run the protocol. Once corrupted, a party becomes byzantine, which means that it can arbitrarily deviate from the protocol under the adversary's control. If a party is never corrupted by the adversary, then we call it honest. We do not assume secure erasures, which means that if the adversary corrupts a party, then it learns everything the party ever knew.

The adversary schedules messages as it wants with arbitrarily long delays, only subject to the rule that it has to eventually deliver every message whose sender is honest. The adversary is strongly adaptive, which means that if any party sends a message, then the adversary can adaptively corrupt the party and then drop the message instead of delivering it.

Our transformation itself works in the full-information setting, where the adversary is omniscient. However, as this setting forbids the strong coins that the transformation requires \cite{hk26}, we obtain our common coins $\C_\mathsf{perfect}$ and $\C_\mathsf{crypto}$ by assuming a secure channel setup. That is, on top of assuming that the parties are connected via fully reliable and authenticated channels, we assume that the adversary cannot read a party' internal state without corrupting the party, and that it cannot read any message in transit without corrupting its sender or its recipient. Note that the message privacy assumption can be made without loss of generality in the cryptographic setting thanks to key exchange protocols \cite{dh76} and symmetric encryption.

Our transformed common coin protocol $\mathit{T}(\C_\str,z,k,\varepsilon,\alpha)$ achieves liveness, which means that if the honest parties all run some common instance of the protocol forever, then they all output from it. A party cannot stop running (terminate) when it outputs because this could lead to the other parties never obtaining outputs. If this is insufficient and termination is required, one can use the termination procedure in \cite{mw25}, which can terminate any binary coin with $O(n^2)$ added constant-size messages and $3$ added latency.

We define protocol latency based on \cite{aw24}. We imagine that every event $e$ that happens in a protocol's execution has a timestamp $t_e$ given by an external global clock. The clock respects causality, which means if an event $e$ must logically have happened before another event $e'$, then $t_e < t_{e'}$. If a protocol's latency is $R$, then the following is true with respect to the clock: \linebreak If by some time $T$ every honest party has become active (has begun sending messages and has acquired an input if required), then by the time $T + R\Delta$ the protocol will have satisfied its liveness requirement, where $\Delta$ is the maximum amount of time the adversary ever waits to deliver a message whose sender is honest before the requirement is satisfied. For a common coin and for BA, the requirement is the pure (independent-of-safety) one that every honest party outputs, though this will not be the case for our bit publishing protocol.

\section{Bit Publishing}

In this section, we present the protocol through which the committees publish the bits they generate. Let $Q$ be a committee with less than $|Q|/3$ byzantine parties in it such that each honest party $i \in Q$ has an input bit $b_i$. The protocol $\mathit{Publish}(Q, d)$ below allows the parties in $Q$ to collectively publish their input bits to the parties outside $Q$. We primarily design it for when the honest parties in $Q$ have some common input bit $b_Q$. When this is the case, the protocol $\mathit{Publish}(Q, d)$ ensures for some set $P^* \subseteq [n]$ of more than $n - d$ parties that every honest party in $P^*$ outputs $b_Q$, indicating the party's belief that the committee $Q$ published this bit. However, even when the honest parties in $Q$ have conflicting inputs, $\mathit{Publish}(Q, d)$ guarantees that the honest parties in $P^*$ all output from $\mathit{Publish}(Q, d)$. Note that the set $P^*$ \linebreak is not fixed in advance; it depends on who the adversary corrupts.

\subparagraph{Crusader Agreement.} The first step of $\mathit{Publish}(Q, d)$ is that the parties in $Q$ run a binary crusader agreement protocol (which we will refer to as $\mathsf{Crusader}$) among themselves to reach crusader agreement on their inputs. Crusader agreement \cite{d82, crusader-modern, aw24} is a well-known weaker variant of byzantine agreement where the parties may output a special value $\bot$ outside the protocol's input domain ($\{0, 1\}$ for binary crusader agreement) if they have no common input. If the honest parties run a crusader agreement protocol with any common input $v$, then they must all output $v$, as in byzantine agreement. However, if their inputs differ, then they only need to reach weak agreement, which means that there should exist some $v$ such that they all output either $v$ or $\bot$. The permission for the extra output $\bot$ means that $n$ parties can reach crusader agreement with $O(n^2)$ constant-size messages and $3$ latency when less than a third of them are byzantine (via the protocol in \cite{mw25}, which calls this task 1-graded consensus).

\subparagraph{Sparse Communication.} In $\mathit{Publish}(Q, d)$, we use a bipartite graph $G_{Q, d}$ with the partite sets $Q$ and $V = \{v_1, \dots, v_n\}$ such that a party $i \in Q$ sends its $\mathsf{Crusader}$ output to a party $j \in [n]$ iff $i \in Q$ and $v_j \in V$ are adjacent in $G_{Q, d}$. Naturally, to minimize communication we want $G_{Q, d}$ to be sparse, and weakening $\mathit{Publish}(Q, d)$'s guarantees by increasing $d$ should permit more sparsity. We formalize this with Lemma \ref{sparselemma}, which says that it suffices for each vertex $v_i \in V$ in $G_{Q, d}$ to have $\Delta = \min(\ceil{\frac{2|Q|}{3}}, \ceil{30(\frac{|Q| \ln 2}{d} + \ln n)})$ neighbors in $Q$.

\apxthm{lemma}{sparselemma}{
    For all integers $n, d \geq 1$ and all $s \in [n]$, there exists a bipartite graph $(Q,V,E)$ whose two partite sets $Q$ of size $s$ and $V$ of size $n$ are such that every vertex in $V$ has exactly $\Delta = \min(\ceil{\frac{2s}{3}}, \ceil{30(\frac{s\ln 2}{d} + \ln n)})$ neighbors in $Q$, and for every set $B \subseteq Q$ of size less than $\frac{s}{3}$ \linebreak there are less than $d$ vertices in $V$ with at least $\frac{\Delta}{2}$ neighbors in $B$.
}

We prove Lemma \ref{sparselemma} in the \hyperref[sparselemma*]{appendix} using the probabilistic method. That is, given any $n$, $s$ and $d$, we construct a bipartite graph by choosing $\Delta$ random neighbors (without repetition) in $Q$ for each vertex in $V$, and show that this graph is satisfactory except with $2^{-\Theta(s)} < 1$ probability. This implies that there exists a graph with the desired properties.

In $\mathit{Publish}(Q, d)$, the parties use a bipartite graph $G_{Q, d} = (Q, V, E)$ which is as described in Lemma \ref{sparselemma}. We assume that $G_{Q, d}$'s vertices are labeled such that its partite set $Q$ is the publishing committee $Q \subseteq [n]$ of $\mathit{Publish}(Q, d)$. The other partite set $V = \{v_1, \dots, v_n\}$ is a copy of $[n]$, relabeled so that it does not overlap with $Q$.

\begin{algobox}{$\mathit{Publish}(Q, d)$}
    \underline{Code for a party $i \in Q$:}
    \begin{algorithmic}[1]
    \State run a common instance of $\mathsf{Crusader}$ with the parties in $Q$, where your input is $b_i$
    \State upon outputting some $y_i$ from $\mathsf{Crusader}$, send $y_i$ to every party $j \in [n] \setminus Q$ such that $i$ and $v_j$ are adjacent $G_{Q, d}$
    \State output $b_i$ \llabel{publish-direct}
    \algstore{publish} \end{algorithmic}
    \underline{Code for a party $i \not \in Q$}
    \begin{algorithmic}[1] \algrestore{publish}
        \State wait to receive either $b$ or $\bot$ from more than $\frac{\Delta}{2}$ parties in $N(v_i)$ (the neighborhood of $v_i$ in $G_{Q, d}$) for any bit $b$, and output $b$ if this happens
    \end{algorithmic}
\end{algobox}

\begin{lemma} \label{publemma}
    Let $Q \subseteq [n]$ be any committee whose honest members $i$ have input bits $b_i$. For any $d \geq 1$, if the parties $1,\dots,n$ run $\mathit{Publish}(Q, d)$ against a byzantine adversary and this adversary does not corrupt $\frac{|Q|}{3}$ or more members of $Q$, then there is some execution-dependent set $P^*$ of more than $n - d$ parties such that 1) every honest party in $P^*$ outputs a bit from $\mathit{Publish}(Q, d)$, and 2) if the honest members of $Q$ have a common input $b_Q$, then every honest party in $P^*$ outputs $b_Q$ from $\Pub(Q, d)$.
\end{lemma}

\begin{proof}
    Suppose the adversary corrupts less than a third of $Q$. Then, the honest parties in $Q$ safely reach crusader agreement on their inputs, and for some bit $b^*$ they all obtain crusader agreement outputs in $\{b^*, \bot\}$. Let $B \subseteq Q$ be the set of parties in $Q$ the adversary adaptively corrupts while the parties run $\mathit{Publish}(Q, d)$, and let $P^* = \{i \in [n] : |N(v_i) \setminus B| > \frac{\Delta}{2}\}$. Every honest party $i \in Q$ directly outputs its input $b_i$ from $\mathit{Publish}(Q, d)$, and every honest party $i \in P^* \setminus Q$ eventually becomes able to output some bit from $\mathit{Publish}(Q, d)$ due to it receiving either $b^*$ or $\bot$ from the $|N(v_i) \setminus B| > \frac{\Delta}{2}$ honest parties in $Q$ who sent their crusader agreement outputs to $i$. Furthermore, if the honest parties in $Q$ have any common input $b_Q$, then they all output $b_Q$ from both $\mathsf{Crusader}$ and from $\mathit{Publish}(Q, d)$, and every honest party $i \in P^* \setminus Q$ outputs $b_Q$ from $\mathit{Publish}(Q, d)$ because it receives $b_Q$ from $|N(v_i) \setminus B)| > \frac{\Delta}{2}$ honest parties in $N(v_i)$ while possibly receiving $\bot$ or $1 - b_Q$ from only the $|N(v_i) \cap B| < \Delta - \frac{\Delta}{2} = \frac{\Delta}{2}$ corrupt parties in $N(v_i)$. Hence, we conclude that every honest party in $Q \cup (P^* \setminus Q) = P^*$ outputs some bit from $\mathit{Publish}(Q, d)$, and that this bit is $b_Q$ if $b_Q$ is the common input of the honest parties in $Q$. This concludes the proof since $G_{Q, d}$ guarantees $|P^*| > n - d$ by definition.
\end{proof}

\subparagraph{Complexity.} In $\mathit{Publish}(Q, d)$, the parties in $Q$ first run $\mathsf{Crusader}$, and this costs the protocol $O(|Q|^2)$ constant-size messages and 3 latency \cite{mw25}. Following this, on average they send their $\mathsf{Crusader}$ outputs to $\frac{|[n] \setminus Q|\Delta}{|Q|}$ parties in $[n] \setminus Q$, with this costing $(n - |Q|)\Delta = O(n(\frac{|Q|}{d} + \log n))$ additional constant-size messages and one more round. All in all, $\mathit{Publish}(Q, d)$'s latency is $4$, and in it the parties send $O(|Q|^2 + n(\frac{|Q|}{d} + \log n))$ constant-size messages in total. The (impure) \linebreak liveness requirement that is satisfied with the latency $4$ is that every honest party $i$ such that \linebreak $|N(v_i) \setminus B| < \frac{\Delta}{2}$ outputs a bit which some honest party in $Q$ has as its input.

\section{The Transformation} \label{sectrans}

In this section, we present the transformed binary coin $\mathit{T}(\C_\str,z,k,\varepsilon,\alpha)$. We begin by stating the full version of Theorem \ref*{simplethm}, rather than the simplified version we gave in Section \ref{sec-contributions}. Note that the full theorem implies the simplified one when $\C_\str$'s message complexity is $\widetilde{O}(n^k)$.

\newcounter{savedtheorem}\setcounter{savedtheorem}{\value{theorem}}
\setcounter{fullthm}{2}
\addtocounter{fullthm}{-1}
\begin{theorem}[Full Version]
\label{fullthm}
    Let $z,k,\varepsilon,\alpha,\delta$ be some positive real parameters such that $z < 1$, $k \geq 2$, $\varepsilon < \alpha \leq \frac{1}{3}$ and $\delta \leq 1$. Suppose we have some $\delta$-fair binary common coin $\C_\str$ such that when $n$ parties run it, $\C_\str$ can tolerate up to $t < \alpha n$ adaptive byzantine faults, and it costs at \linebreak most $M(n)$ messages of size at most $L(n)$ for some functions $M,L$ of $n$. Then, we can obtain a binary common coin $\mathit{T}(\C_\str,z,k,\varepsilon,\alpha)$ with the following properties:  \begin{itemize}
        \item $\mathit{T}(\C_\str,z,k,\varepsilon,\alpha)$ tolerates up to $t \leq (\alpha - \varepsilon)n$ adaptive byzantine faults.
        \item $\mathit{T}(\C_\str,z,k,\varepsilon,\alpha)$ is $(\delta^q(1 - z))$-fair, where $q = 2\ceil{n^{2-2/k}} + 1$.
        \item There is some committee size $s = O(\min(n, \frac{2\alpha - \varepsilon}{z\varepsilon^2}(n^{1/k} + \log n)))$ such that the common coin $\mathit{T}(\C_\str,z,k,\varepsilon,\alpha)$ costs $O(n^{2-2/k}M(s))$ messages of size $L(s) + O(\log n)$ plus $O(\frac{sn^{3 - 3/k}}{z(1 - 3\alpha + 3\varepsilon)} + n^{2-2/k}(s^2 + n\log n))$ more messages of size $O(\log n)$ when $n$ parties run it, and such that if $\C_\str$'s latency is at most a positive integer $R$ when $s$ parties run it, then $\mathit{T}(\C_\str,z,k,\varepsilon,\alpha)$'s latency is at most $R + 5$ when $n$ parties run it.
    \end{itemize}
\end{theorem}
\setcounter{theorem}{\value{savedtheorem}}

The common coin $\mathit{T}(\C_\str,z,k,\varepsilon,\alpha)$ most fundamentally depends on Lemma \ref{comlemma}, which says that for any $q, c \in \mathbb{Z}_{> 0}$, if we have $q$ committees of size $s = \min(n, \ceil{(\frac{2\alpha - \varepsilon}{\varepsilon^2})(\frac{n \ln 2}{c} + \ln q)})$ in $\mathit{T}(\C_\str,z,k,\varepsilon,\alpha)$, then we will have less than $c$ bad committees.

\apxthm{lemma}{comlemma}{
    For all integers $n,q \geq 1$, all $c \in [q]$ and all positive reals $\varepsilon,\alpha$ so that $\varepsilon < \alpha \leq \frac{1}{3}$, there exists a list of $q$ sets $Q_1, \dots, Q_q \subseteq [n]$, each of size $s = \min(n, \ceil{(\frac{2\alpha - \varepsilon}{\varepsilon^2})(\frac{n \ln 2}{c} + \ln q)})$, such that for every set $B \subseteq [n]$ of size at most $(\alpha - \varepsilon)n$, the number of sets $Q_i$ which satisfy $|Q_i \cap B| \geq \alpha s$ is less than $c$.
}

Lemma \ref{comlemma} is related to Bracha's Lemma 7 in \cite{bracha87}, where he considered the case $q = \Theta(n^2)$, $c = \Theta(n)$, $s = \Theta(\log n)$ and $\alpha = \frac{1}{3}$. We prove Lemma \ref{comlemma} in the \hyperref[comlemma*]{appendix} via the probabilistic method, the same way we prove Lemma \ref{sparselemma}. That is, we let $Q_1, \dots, Q_q$ be independent random $s$-element subsets of $[n]$, and show that with $1 - 2^{-\Theta(n)} > 0$ probability the sets $Q_1, \dots, Q_q$ are satisfactory. Since we do not want to assume that $\C_\str$ supports party repetitions, we use set committees which cannot contain repeated parties rather than multiset committees which can. In the lemma's proof, this causes the number of corrupt parties in a random committee to be distributed hypergeometrically rather than binomially. This is not an issue as one can \linebreak treat a hypergeometric variable as a sum of independent Bernoulli variables \cite{hp14}.

Another lemma we need is Lemma \ref{binomlemma}, an anti-concentration bound for symmetric binomial random variables. The symmetric binomial random variable that we care about is the sum of the committees' random bits, and we will use the lemma to lower bound the probability that this sum is at least a certain distance away from its expectation. This will allow us to lower \linebreak bound $\mathit{T}(\C_\str,z,k,\varepsilon,\alpha)$'s fairness probability.

\begin{lemma} \label{binomlemma}
    For any $n \geq 1$, let $X = \sum_{i = 1}^nb_i$ for some i.i.d.\ Bernoulli random variables $b_1, \dots, b_n$ such that $\Pr[b_i = 0] = \Pr[b_i = 1] = \frac{1}{2}$ for all $i \in [n]$. Then, the random variable $X$ obeys the anti-concentration inequality $\Pr[|X - \frac{n}{2}| < \sigma] \leq \frac{8\sigma}{5\sqrt{n}}$ for all integers $\sigma \geq 0$.
\end{lemma}

\begin{proof}
    The binomial random variable $X$ is distributed such that that for all $k \in \{0, \dots, n\}$ one has $\Pr[X = k] = \binom{n}{k}2^{-n} = \frac{n!\,2^{-k}}{k!\,(n-k)!} = \frac{\Gamma(n+1)2^{-n}}{\Gamma(k + 1)\Gamma(n - k + 1)} \leq \frac{\Gamma(n+1)2^{-n}}{(\Gamma(n/2 + 1)))^2}$. The last inequality here is due to the log-convexity of the $\Gamma$ function \cite{bohr1920laerebog}, which means that $\Gamma(a)\Gamma(b) \geq \Gamma(\frac{a+b}{2})$ for all positive reals $a,b$. It is known that $\frac{\Gamma(n+1)}{(\Gamma(n/2 + 1))^2} = \frac{\Gamma(n/2+0.5)2^n}{\Gamma(n/2 + 1)\sqrt{\pi}}$ \cite{elezovic14}, and that $\frac{\Gamma(n/2+x)}{\Gamma(n/2+1)} \leq (n/2)^{x-1}$ when $0 \leq x \leq 1$ \cite{gautschi59}. Finally, when we put all these inequalities together, we get $\binom{n}{k}2^{-n} \leq \frac{\Gamma(n+1)2^{-n}}{\Gamma(k + 1)\Gamma(n - k + 1)} \leq \frac{\Gamma(n+1)2^{-n}}{(\Gamma(n/2 + 1)))^2} = \frac{\Gamma(n/2+0.5)}{\Gamma(n/2+1)\sqrt{\pi}} \leq \frac{(n/2)^{-0.5}}{\sqrt{\pi}} < \frac{4}{5\sqrt{n}}$.

Above, we prove $\Pr[X = k] < \frac{4}{5\sqrt{n}}$ for all $k$. Using this bound and the fact that $|X - \frac{n}{2}| < \sigma$ holds for at most $2\sigma$ values $X$ can attain, we get $\Pr[|X - \frac{n}{2}| < \sigma] = \Sigma_{k \in \mathbb{Z}\,:\,|k - n/2| < \sigma}\Pr[X = k] \leq (2\sigma)\frac{4}{5\sqrt{n}} = \frac{8\sigma}{5\sqrt{n}}$ for all integers $\sigma \geq 0$.
\end{proof}

Finally, we are ready for $\mathit{T}(\C_\str, z, k, \varepsilon, \alpha)$. It works as explained in Section \ref{sec-contributions}. One small detail which we skipped before is that the number of committees $q$ should be odd. Otherwise, $\mathit{T}(\C_\str, z, k, \varepsilon, \alpha)$ would only achieve $(1 - O(q^{-0.5}))$-fairness even if the committees were all good and they all generated and published their bits without any problems at all, because of the tie-breaking required when exactly 50\% of the committees generate each bit. To avoid a $O(q^{-0.5})$ fairness loss, we choose $q$ to be odd, and we use an adjusted fairness loss parameter $z' = \max(\frac{z}{3}, z - 1.6q^{-0.5})$ in $\mathit{T}(\C_\str, z, k, \varepsilon, \alpha)$ instead of $z$.

\begin{algobox}{$\mathit{T}(\C_\str, z, k, \varepsilon, \alpha)$}
    \underline{Initialization:}
    \begin{algorithmic}[1]
        \State let $q = 2\ceil{n^{2 - 2/k}} + 1$ \Comment{the number of committees}
        \State let $z' = \max(\frac{z}{3}, z - 1.6q^{-0.5})$ \Comment{the adjusted fairness loss parameter}
        \State let $c = \ceil{\frac{z'\sqrt{q}}{3}}$ \Comment{the maximum number of bad committees, plus 1}
        \State let $s = \min(n, \ceil{(\frac{2\alpha - \varepsilon}{z'\varepsilon^2})(\frac{n \ln 2}{c} + \ln q)})$ \Comment{the size of each committee}
        \State let $d = \ceil{\frac{z'n(1 - 3\alpha +3\varepsilon)}{36\sqrt{q}}}$ \Comment{the fault parameter for the $\Pub$ instances}
        \State let $Q_1, \dots, Q_q$ be the $q$ publicly known committees of size $s$ such that for any set $B \subseteq [n]$ of size at most $(\alpha - \varepsilon)n$, we have $|\{i \in [q] : |Q_i \cap B| \geq \alpha s\}| < c$.
    \algstore{transform} \end{algorithmic}
    \underline{Code for a party $i$:}
    \begin{algorithmic}[1] \algrestore{transform}
        \State let $v_0, v_1, w_0, w_1 \gets 0, 0, 0, 0$
        \State let $Q_i^* = \{j \in [q] : i \in Q_j\}$ \Comment{the indices of the committees that contain $i$}
        \State for every $j \in Q_i^*$, run an instance of $\C_\str$ (named $\C_\str^j$) with the parties in $Q_j$
        \State for every $j \in [q]$, run $\mathit{Publish}(Q_j, d)$
        \Upon{outputting $b$ from $\C_\str^j$ for any $j \in Q_i^*$}
            \State let $b$ be your input in $\mathit{Publish}(Q_j, d)$
        \EndUpon
        \Upon{outputting $b$ from $\mathit{Publish}(Q_j, d)$ for any $j \in [q]$}
            \State let $v_b \gets v_b + 1$
            \If{$v_0 + v_1 = q - \floor{\frac{5z'\sqrt{q}}{12}}$}\llabel{live-condition}
                \State send $b_\mathsf{maj}$ to every party, where $b_\mathsf{maj} = 0$ if $v_0 \geq v_1$ and $b_\mathsf{maj} = 1$ otherwise\llabel{bcast-majority}
            \EndIf
        \EndUpon
        \Upon{receiving the first bit $b$ that the party $j$ sent you for any $j \in [n]$}
            \State let $w_{b} \gets w_{b} + 1$\llabel{output-condition}
            \If{$w_0 + w_1 = \floor{\frac{2n}{3}} + 1$}
                \State output $0$ if $w_0 \geq w_1$, and output $1$ otherwise\llabel{output-line}
            \EndIf
        \EndUpon
    \end{algorithmic}
\end{algobox}

\begin{proof}[Proof of Theorem \ref{fullthm}] Let $B$ be the set of parties the adversary adaptively corrupts during the execution of $\mathit{T}(\C_\str, z, k, \varepsilon, \alpha)$, and suppose $|B| \leq t \leq (\alpha - \varepsilon)n$ as $\mathit{T}(\C_\str, z, k, \varepsilon, \alpha)$ only has to tolerate $t \leq (\alpha - \varepsilon)n$ faults. The number of good committees (committees $Q_i$ such that $|Q_i \cap B| < \alpha s$) is at least $q - c + 1 > q - \frac{z'\sqrt{q}}{3}$. For any good committee $Q_i$, the proportion of corrupt parties in $Q_i$ is less than $\alpha \leq \frac{1}{3}$; so, $\C_\str^i$ and $\Pub(Q_i, d)$ both work well.

Since the coin $\C_\str$ is $\delta$-fair, we can associate each committee $Q_i$ with two freshly sampled random bits $b_i$ and $g_i$ such that $b_i$ is uniformly random and $\Pr[g_i = 1] = \delta$. If $Q_i$ is a bad committee which contains at least $\alpha s$ corrupt parties, then these bits mean nothing. However, if $Q_i$ is a good committee, then the honest parties in $Q_i$ all output from $\C_\str^i$, and in this case $g_i = 1$ suffices for their outputs from $\C_\str^i$ to all equal $b_i$. This is a valid description of how $\C_\str^i$ behaves since $g_i = 1$ indicates the good event that the honest parties in $Q_i$ are required to all output some freshly sampled uniformly random bit from the instance (that we call $b_i$) from $\C_\str^i$, and $\Pr[g_i = 1] = \delta$ probability by the definition of $\delta$-fairness.

Let us say that an honest party $i$ mishears some committee $Q_j$ if $Q_j$ is a good committee, but the party $i$ either outputs $1 - b_j$ from $\Pub(Q_j, d)$ even though $g_j = 1$, or it outputs nothing from $\Pub(Q_j, d)$. Accordingly, let $H^*$ be the set of honest parties who mishear less than $\frac{z'\sqrt{q}}{12}$ committees.

\begin{lemma} \label{hsize-lemma}
    The set $H^*$ contains at least $\floor{\frac{2n}{3}} + 1$ parties.
\end{lemma}

\begin{proof}
    Let $Q_i$ be a good committee. The honest parties in $Q_i$ all output from $\C_\str^i$ and hence all obtain $\Pub(Q_i, d)$ inputs, which means that for some set $P^*_i$ of at least $n - d + 1 > n - \frac{z'n(1 - 3\alpha + 3\varepsilon)}{36\sqrt{q}}$ parties, the honest parties in $P^*_i$ all output from $\Pub(Q_i, d)$. Furthermore, if $g_i = 1$, then the honest parties in $Q_i$ all output the bit $b_i$ from $\C_\str^i$ and subsequently let their $\Pub(Q_i, d)$ inputs be $b_i$, and this means that the honest parties in $P^*_i$ output $b_i$ from \linebreak $\Pub(Q_i, d)$. In short, every honest party $j \in [n]$ that mishears $Q_i$ is in $[n] \setminus P_i^*$, which is of size less than $\frac{z'n(1 - 3\alpha + 3\varepsilon)}{36\sqrt{q}}$.

    Since every committee is misheard by less than $\frac{z'n(1 - 3\alpha + 3\varepsilon)}{36\sqrt{q}}$ honest parties, the number of honest parties outside $H^*$ must be less than $\frac{n(1 - 3\alpha + 3\varepsilon)}{3}$. Otherwise, the average number of mishearings per committee would be at least $\frac{1}{q} \cdot \frac{z'\sqrt{q}}{12} \cdot \frac{n(1 - 3\alpha + 3\varepsilon)}{3} = \frac{z'n(1 - 3\alpha + 3\varepsilon)}{36\sqrt{q}}$, contradicting our previous observation that every committee is misheard less often than this. Finally, we have $|H^*| > n - t - \frac{n(1 - 3\alpha + 3\varepsilon)}{3} = \frac{2n}{3} - t + (\alpha - \varepsilon)n \geq \frac{2n}{3}$. \qedhere
\end{proof}

\begin{lemma} \label{hgood-lemma}
    There exists an event that occurs with at least $\delta^q(1 - z)$ probability such that if this event occurs, then there is some bit $b^*$ such that the parties in $H^*$ can only broadcast $b^*$ on Line \ref{line:bcast-majority}, and $b^*$ is uniformly random given that this event occurs.
\end{lemma}

\begin{proof}
    Suppose $\bigwedge_{i=1}^q(g_i = 1)$ and $|X - \frac{q}{2}| \geq \frac{5z'\sqrt{q}}{8}$, where $X = \sum_{i = 1}^qb_i$. Let $b^*$ be the bit indicating if $X > \frac{q}{2}$, with $b^* = 1$ if $X > \frac{q}{2}$ and $b^* = 0$ otherwise. If an honest party $i$ outputs $1 - b^*$ from $\Pub(Q_j, d)$ for a committee $Q_j$, then either $Q_j$ is bad, or $b_j = 1 - b^*$, or $i$ has misheard $Q_j$. So, for every $i \in H^*$, we have $|\{j \in [q] : i \text{ output } 1 - b^* \text{ from } \mathit{Publish}(Q_j, d)\}| \leq |\{j \in [q] : Q_j \text{ is bad}\}| + |\{j \in [q] : b_j = 1 - b^*\}| + |\{j \in [q] : i \text{ misheard } Q_j\}| < \frac{z'\sqrt{q}}{3} + (\frac{q}{2} - \frac{5z'\sqrt{q}}{8}) + \frac{z'\sqrt{q}}{12} = \frac{q}{2} - \frac{5z'\sqrt{q}}{24}$. This means that if some party $i \in H^*$ outputs from $q - \floor{\frac{5z'\sqrt{q}}{12}}$ instances of $\mathit{Publish}$ and thus observes that $v_0 + v_1 = q - \floor{\frac{5z'\sqrt{q}}{12}}$ on Line \ref{line:live-condition}, then it does so with $v_{b^*} > v_{1-b^*}$ (since $v_{b^*} - v_{1-b^*} = q - \floor{\frac{5z'\sqrt{q}}{12}} - 2v_{1-b^*} > q - \floor{\frac{5z'\sqrt{q}}{12}} - 2(\frac{q}{2} - \frac{5z'\sqrt{q}}{24}) \geq 0$), and therefore it broadcasts $b^*$ on Line \ref{line:bcast-majority}.

    Next, we show that $\Pr[\bigwedge_{i=1}^q(g_i = 1) \land |X - \frac{q}{2}| \geq \frac{5z'\sqrt{q}}{8}] \geq \delta^q(1 - z)$. By the independence of $g_1,\dots,g_q$, we have $\Pr[\bigwedge_{i=1}^q(g_i = 1)|] = \delta^q$, independently of whether $|X - \frac{q}{2}| \geq \frac{5z'\sqrt{q}}{8}$ since the $g_i$ and $b_i$ vectors are independent. So, it only remains to show that $\Pr[|X - \frac{q}{2}| < \frac{5z'\sqrt{q}}{8}] \leq z$. Lemma \ref{binomlemma} gives us $\Pr[|X - \frac{q}{2}| < \frac{5z'\sqrt{q}}{8}] \leq \frac{8\ceil{(5z'\sqrt{q})/8}}{5\sqrt{q}} < \frac{8((5z'\sqrt{q})/8 + 1)}{5\sqrt{q}} = z' + 1.6q^{-0.5}$. If $z' = z - 1.6q^{-0.5}$, then we are done. If not, then because $z' = \max(\frac{z}{3}, z - 1.6q^{-0.5})$, we have $z' = \frac{z}{3} > z - 1.6q^{-0.5}$. This means that $z' < 0.8q^{-0.5}$, which gives us $\Pr[|X - \frac{q}{2}| < \frac{5z'\sqrt{q}}{8}] \leq \Pr[|X - \frac{q}{2}| < 0.5] = 0 < z$. The fact that $\Pr[|X - \frac{q}{2}| < 0.5] = 0$ follows from $q$ being odd.

    Finally, $b^*$ is uniformly random given $\bigwedge_{i=1}^q(g_i = 1) \land |X - \frac{q}{2}| \geq \frac{5z'\sqrt{q}}{8}$, by $X$'s symmetry and by the independence of $b_1,\dots,b_q$ from $g_1,\dots,g_q$. This concludes the proof.
\end{proof}

Now, we use Lemma \ref{hsize-lemma} and \ref{hgood-lemma} to lower bound $\mathit{T}(\C_\str, z, k, \varepsilon, \alpha)$'s fairness. Suppose there is some uniformly random bit $b^*$ such that the parties in $H^*$ can only broadcast $b^*$ on Line \ref{line:bcast-majority}. Then, if any honest party outputs from $\mathit{T}(\C_\str, z, k, \varepsilon, \alpha)$ on Line \ref{line:output-line} after receiving bits from $\floor{\frac{2n}{3}} + 1$ parties, it does so with only at most $n - |H^*| < \frac{n}{3}$ (a strict minority) of these bits being $1 - b^*$, which means that the party outputs $b^*$. Since Lemma \ref{hgood-lemma} says that there is such a uniformly random bit $b^*$ with probability at least $\delta^q(1 - z)$, with at least this probability the honest parties' outputs from $\mathit{T}(\C_\str, z, k, \varepsilon, \alpha)$ all equal some common uniformly random bit $b^*$, or in other words $\mathit{T}(\C_\str, z, k, \varepsilon, \alpha)$ is $(\delta^q(1 - z))$-fair as Theorem \ref{fullthm} demands.

Next, we prove $\mathit{T}(\C_\str, z, k, \varepsilon, \alpha)$'s liveness and bound its latency. Let us scale all event timestamps so that the adversary takes at most $1$ unit of time to deliver any message. Let $T$ be any time such that by the time $T$, the honest parties in each good committee $Q_i$ output from $\C_\str^i$ and thus begin publishing their $\C_\str^i$ outputs via $\Pub(Q_i, d)$. Following this, because $\Pub$'s latency $4$, by the time $T + 4$ every honest party outputs from $\Pub(Q_i, d)$ for every good committee $Q_i$ which it does not mishear. The parties in $H^*$ who mishear less than $\frac{z'\sqrt{q}}{12}$ committees output from more than $q - c + 1 - \frac{z'\sqrt{q}}{12} > q - \frac{z'\sqrt{q}}{3} - \frac{z'\sqrt{q}}{12} = q - \frac{5z'\sqrt{q}}{12}$ $\Pub$ instances by the time $T + 4$, which is enough for them to all broadcast bits on Line \ref{line:bcast-majority}. Finally, the parties in $H^*$ broadcasting bits ensures that every honest party $i \in [n]$ receives $|H^*| \geq \floor{\frac{2n}{3}} + 1$ bits by the time $T + 5$, which is enough for $i$ to output from $\mathit{T}(\C_\str, z, k, \varepsilon, \alpha)$ on Line \ref{line:output-line}. In conclusion, we observe every honest party outputs from $\mathit{T}(\C_\str, z, k, \varepsilon, \alpha)$ in at most $5$ time units after every honest member of every good committee $Q_i$ outputs from $\C_\str^i$. This means that $\mathit{T}(\C_\str, z, k, \varepsilon, \alpha)$ achieve liveness, and that if $\C_\str$'s latency is bounded by $R$ if $s$ parties run it (which entails the latency bound $R$ on $\C_\str^i$ for every good committee $Q_i$), \linebreak then $\mathit{T}(\C_\str, z, k, \varepsilon, \alpha)$'s latency is bounded by $R + 5$.

Finally, let us consider the message and communication complexities. Suppose as Theorem \ref{fullthm} stipulates that $\C_\str$ costs at most $M(s)$ messages of size at most $L(s)$ when $s$ parties run it. In $\mathit{T}(\C_\str, z, k, \varepsilon, \alpha)$, we have $q$ instances of $\C_\str$ (each executed by $s$ parties), $q$ $\Pub$ instances $\Pub(Q_1,d), \dots, \Pub(Q_q,d)$, and finally the bit broadcasts on Line \ref{line:bcast-majority}. Each instance of $\C_\str$ costs at most $M(s)$ messages of size $L(s) + O(\log q)$, where the $O(\log q) = O(\log n)$ term is due to the protocol ID tags that must be used so that the parties are able to tell to which $\C_\str$ instance each message belongs. So, the $\C_\str$ instances cost at most $qM(s)$ messages of size $L(s) + O(\log n)$ in total. Then, each instance of $\Pub$ costs $O(s^2 + n(\frac{s}{d} + \log n))$ messages, of size $O(\log q) = O(\log n)$ where again this $O(\log n)$ term is due to the protocol ID tags. This means that the $q$ $\Pub$ instances cost $O(qs^2 + qn(\frac{s}{d} + \log n)) = O(n^{2-2/k}s^2 + n^{3-2/k}(\frac{s\sqrt{q}}{z'n(1-3\alpha + 3\varepsilon)} + \log n)) = O(\frac{sn^{3 - 3/k}}{z(1 - 3\alpha + 3\varepsilon)} + n^{2-2/k}(s^2 + n\log n))$ messages of size $O(\log n)$ in total. Finally, each party might broadcast a bit on Line \ref{line:bcast-majority}, and these broadcasts cost $O(n^2)$ \linebreak constant-size messages in total. Summing everything up, in $\mathit{T}(\C_\str, z, k , \varepsilon, \alpha)$ there are in total $qM(s) = O(n^{2-2/k}M(s))$ messages of size $L(s) + O(\log n)$ plus $O(\frac{sn^{3 - 3/k}}{z(1 - 3\alpha + 3\varepsilon)} + n^{2-2/k}(s^2 + n\log n))$ further messages of size $O(\log n)$, as Theorem \ref{fullthm} demands.
\end{proof}

\subparagraph*{A note on liveness.} To keep things simple, we have only considered common coin protocols with guaranteed liveness, where the honest parties never fail to output. However, we could relax this guarantee by only requiring the honest parties to all output with probability at least $\lambda$, and by calling a common coin $\lambda$-live if it achieves this guarantee. Such a probabilistic guarantee is for instance considered in \cite{cr93}. Our transformation could then support $\lambda$-live coins as follows: If $\C_\str$ is $\lambda$-live, then $\mathit{T}(\C_\str, z, k, \varepsilon, \alpha)$ is $\lambda^q$-live. This is because with at least $\lambda^q$ probability no good committee's strong coin instance would stall (by the independence of the liveness events), and $\mathit{T}(\C_\str, z, k, \varepsilon, \alpha)$ would thus ensure that every honest party outputs.

\section{Using the Transformation} \label{strong-coin-chapter}

While our main contribution is our transformation, we would also like to get some concrete results out of it. To use the transformation, we need suitable strong common coins, which we get by using the Monte Carlo common coin of \cite{fkt22} as a framework. This framework is based on AVSS, a well-known primitive invented for common coin tossing in \cite{cr93} where a dealer party shares a secret value with the other parties in a way that prevents the adversary from learning the secret before the honest parties run a reconstruction protocol to learn it, with some liveness and consistency guarantees that must hold even when the dealer is corrupted. The framework can tolerate up to $t < \frac{n}{3}$ faults, which is optimal for asynchronous byzantine agreement (even though asynchronous coin tossing is possible against $t < \frac{n}{2}$ faults \cite{cks00}). For any $\delta \in (0, 1)$, it produces a $\delta$-fair binary coin using the following ingredients: \begin{itemize}
    \item $n$ parallel instances of AVSS,\footnote{To be precise, what one needs is an AwVSS scheme, as defined in \cite{hashrand}. AwVSS is a subtly weaker variant of AVSS which \cite{hashrand} defines and uses for Monte Carlo coin tossing (in a way that supports amortizing the cost of tossing many coins).} one for each party to share a secret of size $O(\log \frac{1}{1 - \delta})$;
    \item one instance of binding gather;\footnote{Binding gather is a communication primitive that \cite{binding-gather} invented for distributed key generation.}
    \item $n$ parallel instances of approximate agreement in $[0, 1]$, each with the error $\varepsilon = O(\frac{1 - \delta}{n})$.\footnote{Approximate agreement in $[0, 1]$ is a variant of byzantine agreement where the honest parties have inputs in $[0, 1]$, and they must obtain outputs which 1) differ by at most some error parameter $\varepsilon > 0$ from each \linebreak other, and 2) are in between the minimum and maximum honest inputs \cite{approx-1986}.}
\end{itemize}

The Monte Carlo common coin in \cite{fkt22} uses specific AVSS, binding gather and approximate agreement protocols. However, since these protocols are used in a black-box way, we can replace them with alternatives. This is what we mean by using the coin as a framework. 

For binding gather, we use the protocol from \cite{aaps24}. It costs $O(n^3)$ messages of size $O(\log n)$ and has a constant latency. It tolerates $t < \frac{n}{3}$ faults with perfect security.

For approximate agreement, we use the protocol from \cite{mw25}, which also tolerates $t < \frac{n}{3}$ faults with perfect security. The cost for $n$ parallel instances of this protocol is $O(n^3\log\frac{1}{\varepsilon}) = O(n^3\log\frac{n}{1 - \delta})$ messages of size $O(\log n + \log\log\frac{1}{\varepsilon}) = O(\log n + \log\log\frac{n}{1-\delta})$ and $O(\log \frac{1}{\varepsilon}) = O(\log\frac{n}{1-\delta})$ latency.\footnote{In \cite{fkt22}, the authors bundle the $n$ approximate agreement instances together to keep the communication cubic in $n$. With the approximate agreement protocol from \cite{mw25}, bundling is no longer needed.}

The tricky ingredient is AVSS. Since AVSS involves input secrecy, it is impossible in the full-information setting, unlike binding gather and approximate agreement. So, we limit the adversary by assuming secure channels. To get $\C_\mathsf{perfect}$ we use a perfectly secure AVSS scheme that only tolerates $t < \frac{n}{4}$ faults, while to get $\C_\mathsf{crypto}$ we use a more efficient cryptographically secure scheme that tolerates $t < \frac{n}{3}$ faults.

As $\mathit{T}(\C_\str, z, k, \varepsilon, \alpha)$ is $(\delta^q(1 - z))$-fair (where $q = 2\ceil{n^{2-2/k}} + 1$) when $\C_\str$ is $\delta$-fair and $\delta^q(1 - z) \geq 1 - (1 - \delta)q - z$, for any $\delta' \in (0, 1)$ we can ensure that $\mathit{T}(\C_\str, z, k, \varepsilon, \alpha)$ is $\delta'$-fair by letting $z = \frac{1 - \delta'}{2}$ and letting $\delta = 1 - \frac{1 - \delta'}{2q}$. The case we care about is when $\delta' < 1 - \frac{1}{n}$. If $\delta'$ were any closer to $1$, then we would have committees of size $n$ in $\mathit{T}(\C_\str, z, k, \varepsilon, \alpha)$, and there would thus be no reason to use $\mathit{T}(\C_\str, z, k, \varepsilon, \alpha)$ instead of $\C_\str$.

\subparagraph*{Perfect Security.} To obtain the perfectly secure $\C_\mathsf{perfect}$, we use the perfectly secure AVSS scheme from \cite{aaps24} that tolerates $t < \frac{n}{4}$ faults, costs $O(n^3)$ messages of size $O(\log n)$ and has a constant latency. For any $\delta' < 1 - \frac{1}{n}$, we can use the Monte Carlo framework with this AVSS scheme to obtain a perfectly secure $(1 - \frac{1 - \delta'}{4\ceil{n^{1.5}} + 2})$-fair common coin $\C_\str^\mathsf{perfect}$ which tolerates $t < \frac{n}{4}$ faults, costs $O(n^4)$ messages of size $O(\log n)$ and has $O(\log n)$ latency. For $z = \frac{1 - \delta'}{2}$, $k = 4$, $\alpha = \frac{1}{4}$ and any $\varepsilon > 0$, the transformed common coin $\C_\mathsf{perfect} = \mathit{T}(\C_\str^\mathsf{perfect}, z, k, \varepsilon, \alpha)$ is a $\delta'$-fair binary common coin which tolerates $t \leq (\frac{1}{4} - \varepsilon)n$ faults, has $O(\log n)$ latency, and costs $O(n^{2.5}(\varepsilon^{-8}(1 - \delta')^{-4} + \log n))$ messages of size $O(\log n)$. Out of these messages, one can attribute $O(n^{1.5}s^4) = O(z^{-4}\varepsilon^{-8}n^{2.5})$ to the $O(n^{1.5})$ committees of size $s = O(z^{-1}\varepsilon^{-2}n^{0.25})$ each generating random bits, and $O(sn^{2.25} + n^{1.5}(s^2 + n\log n)) = O(z^{-1}\varepsilon^{-2}n^{2.5} + z^{-2}\varepsilon^{-4}n^2 + n^{2.5}\log n)$ to the rest of the protocol.

\subparagraph*{Cryptographic Security.} For cryptographic security, we use the AVSS (or more precisely AwVSS) scheme from \cite{hashrand} which tolerates $t < \frac{n}{3}$ faults, costs $O(n^2)$ messages of size $O(\kappa)$ (where $\kappa = \Omega(\log n)$ is a cryptographic security parameter) and has a constant latency. Then, given any $\delta' < 1 - \frac{1}{n}$, we can use the Monte Carlo framework with this scheme to obtain a $(1 - \frac{1 - \delta'}{4\ceil{n^{4/3}} + 2})$-fair coin $\C_\str^\mathsf{crypto}$ which tolerates $t < \frac{n}{3}$ faults, costs $O(n^3\log n)$ messages of size $O(\kappa)$ and has $O(\log n)$ latency. For $z = \frac{1 - \delta'}{2}$, $k = 3$, $\alpha = \frac{1}{3}$ and any $\varepsilon > 0$, the transformed common coin $\C_\mathsf{crypto} = \mathit{T}(\C_\str^\mathsf{crypto}, z, k, \varepsilon, \alpha)$ is a $\delta'$-fair binary common coin which tolerates $t \leq (\frac{1}{3} - \varepsilon)n$ faults, has $O(\log n)$ latency, and costs $O(n^{7/3}(1 - \delta')^{-3}\varepsilon^{-6}\log n)$ messages of size at most $O(\kappa)$. Out of these messages, $O(n^{4/3}s^3\log s) = O(n^{7/3}\varepsilon^{-6}z^{-3}n^{7/3}\log n)$ are due to the $O(n^{4/3})$ committees of size $s = O(\min(n, z^{-1}\varepsilon^{-2}n^{1/3}))$ each generating random bits, and this dominates the $O(\frac{sn^2}{3\varepsilon} + n^{4/3}(s^2 + n\log n))$ message complexity of the rest of the protocol. \linebreak Note that the AwVSS scheme's liveness properties hold unconditionally; therefore, the strong coin $\C_\str^\mathsf{crypto}$ and thus $\C_\mathsf{crypto}$ also achieve liveness unconditionally. However, the coins' fairness guarantees rely on the cryptographic security of the AwVSS scheme.

\subparagraph*{Multivalued Coins.} For any $\delta$ and $\ell$, one can obtain a $\delta$-good $\ell$-bit coin toss by taking a $(1 - \frac{1 - \delta}{\ell})$-fair binary coin, tossing it $\ell$ times in parallel, and concatenating the bits into a single $\ell$-bit coin output. By the union bound, with probability at least $\delta$ the bits will be i.i.d.\ uniformly random, and thus their concatenation will be uniformly random. This strategy is compatible with $\C_\mathsf{perfect}$ and $\C_\mathsf{crypto}$ as they have adjustable fairness parameters. Furthermore, if $\ell = O(\mathrm{polylog}(n))$, then the communication complexity of obtaining a $\delta$-fair $\ell$-bit coin toss this way using $\C_\mathsf{perfect}$ or $\C_\mathsf{crypto}$ is $O(\mathrm{polylog}(n))$ times that of using $\C_\mathsf{perfect}$ or $\C_\mathsf{crypto}$ as $\delta$-fair binary common coins, while the latency only goes up by $O(\log \ell) = O(\log \log n)$ as the strong coin's latency in this construction is $O(\log \frac{n}{1 - (1 - (1 - \delta)/\ell)}) = O(\log \frac{\ell n}{1 - \delta})$. In particular, one can use $\C_\mathsf{perfect}$ or $\C_\mathsf{crypto}$ for leader election with $O(\mathrm{polylog}(n))$ times the communication complexity of tossing a binary coin with the same fairness, which can be of use if a protocol calls for leader election rather than a binary coin (such as \cite{ddlmrs24} does for core set agreement). For this, the parties can toss an $\ell$-bit coin for some $\ell \gg \log_2 n$, interpret the output as a number in $\{0, \dots, 2^\ell - 1\}$, and let the leader's ID be this number modulo $n$ plus $1$.

\bibliography{refs}

\appendix

\section{Skipped Proofs}

\sparselemma

\begin{proof}
    When $d > n$, the lemma trivially holds since the partite set $V$ of size $n$ is too small for it to fail. Hence, we assume for the rest of the proof that $d \leq n$. In addition, we assume $\Delta < \frac{2s}{3}$, which implies $\Delta = \ceil{30(\frac{s\ln 2}{d} + \ln n)}$. If this were false, then the lemma would hold simply because every vertex $v \in V$ would have at most $|B| < \frac{s}{3} \leq \frac{\Delta}{2}$ neighbors in $B$.
    
    We prove the lemma using the probabilistic method. Let us construct a random bipartite graph by taking two vertex sets $Q$ and $V$, respectively of size $s$ and $n$, and by sampling $\Delta$ uniformly random neighbors in $Q$ without replacement independently for each vertex $v \in V$. This graph witnesses the lemma's truth iff one cannot find a counterexample set $B \subseteq Q$ of size less than $\frac{s}{3}$ such that $d$ vertices in $V$ have at least $\frac{\Delta}{2}$ neighbors in $V$, and we can assume $|B| = \ceil{\frac{s}{3}} - 1$ without loss of generality since letting $B$ be as big as possible cannot make it harder to find a counterexample. That is, no bad event $X_{B,D}$ should occur, where $X_{B,D}$ is the event that $B \subseteq Q$ is a set of size $b = \ceil{\frac{s}{3}} - 1$ and $D \subseteq V$ is a set of size $d$ such that every $v \in D$ has at least $\frac{\Delta}{2}$ neighbors in $B$. Below, we show for every possible choice of $B$ and $D$ that the probability of the bad event $X_{B,D}$ is less than $2^{-s-d\log_2 n}$. This is enough to prove the lemma since there are much fewer than $2^{s + d\log_2 n}$ ways to choose $B$ and $D$ ($\binom{s}{b} = 2^{s - \Theta(s)}$ for $B$ and $\binom{n}{d} \leq n^d = 2^{d\log_2 n}$ for $D$), which by the union bound means that with probability at least $1 - (2^{s - \Theta(s) + d\log_2 n})(2^{-s - d\log_2 n}) = 1 - 2^{-\Theta(s)}$ the bad event $X_{B,D}$ does not occur for \linebreak any $B$ and $D$.

    Let $B \subseteq Q$ be any set of size $b = \ceil{\frac{n}{3}} - 1$, and let $D \subseteq V$ be any set of size $d$. Suppose $b \geq 1$ without loss of generality, as the bad event $X_{B,D}$ is clearly impossible if $B$ is empty. Consider the random variable $Z = \sum_{v \in D}(d_B(v))$, where $d_B(v)$ for each $v \in D$ is the number of neighbors that $v$ has in $B$. For each $v \in V$, $d_B(v)$ is distributed identically to the random number of red balls one would get if one picked $\Delta$ uniformly random balls without replacement from an urn that contains $b$ red balls and $s - b$ blue balls. More formally, $d_b(v) \sim H(s, b, \Delta)$, where $H(s, b, \Delta)$ is the hypergeometric distribution with the parameters $s,b,\Delta$. It is known that the distribution $H(s, b, \Delta)$ is identical to the distribution of the sum of some collection of $\min(b, \Delta)$ independent (but except in degenerate cases not identically distributed) Bernoulli random variables \cite{hp14}, with the expected value $\frac{\Delta b}{s}$ by the linearity of expectation since each of the $\Delta$ picked balls would be red with probability $\frac{b}{s}$.
    
    Since $Z = \sum_{v \in D}(d_B(v))$ is the sum of $d$ i.i.d.\ random variables drawn from the distribution $H(s,b,\Delta)$ (with independence coming from the fact that each $v \in D$ has its neighbors picked independently), one can treat $Z$ as a sum of $d \cdot \min(b, \Delta)$ independent Bernoulli variables with $\mathbb{E}[Z] = \sum_{v \in D}\mathbb{E}(d_B(v)) = \sum_{v \in D}\frac{\Delta b}{s} = \frac{d \Delta b}{s}$, and use Chernoff bounds accordingly.

    Now, observe that if $X_{B, D}$ occurs, then we have $d_B(v) \geq \frac{\Delta}{2}$ for all $v \in D$, which implies $Z = \sum_{v \in D}(d_B(v)) \geq \frac{d \Delta}{2} = (1 + \frac{s - 2b}{2b})\mathbb{E}[Z]$. The Chernoff bound $\Pr[Z \geq (1 + \delta)\mathbb{E}[Z]] \leq \exp(\frac{-\mathbb{E}[Z] \delta^2}{2 + \delta})$ (found in many sources, such as \cite{chernoff}) gives us $-\ln(\Pr[Z \geq \frac{d\Delta}{2}]) = -\ln(\Pr[Z \geq (1 + \frac{s - 2b}{2b}) \mathbb{E}[Z]]) \geq \frac{d\Delta b((s - 2b)/(2b))^2}{s(2 + (s - 2b)/(2b))} = \frac{d\Delta b((s - 2b)/(2b))^2}{s(s + 2b)/(2b)} = \frac{d\Delta(s-2b)^2}{2s(s+2b)} > \frac{d\Delta(s/3)^2}{2s(5s/3)} = \frac{d\Delta}{30} = \frac{d\ceil{30((s\ln 2)/d + \ln n)}}{30} \geq s \ln 2 + d \ln n$. That is, the Chernoff bound ensures that $\Pr[Z \geq \frac{d\Delta}{2}] < e^{-s \ln 2 - d \ln n} = 2^{-s - d\log_2 n}$. Finally, because the bad event $X_{B, D}$ implies $Z \geq \frac{d\Delta}{2}$, we have $\Pr[X_{B,D}] \leq \Pr[Z \geq \frac{d\Delta}{2}] \leq 2^{-s - d\log_2 n}$, as desired.
\end{proof}

\comlemma

\begin{proof}
    If $\ceil{(\frac{2\alpha - \varepsilon}{\varepsilon^2})(\frac{n \ln 2}{c} + \ln q)} > n$, then $s = n$; so, it suffices to set $Q_1 = \dots = Q_q = [n]$, as this ensures that $|Q_i \cap B| = |B| \leq (\alpha - \varepsilon)n < \alpha s$ for every index $i \in [q]$ and every set $B \subseteq [n]$ of size at most $(\alpha - \varepsilon)n$. Hence, we assume below that $s = \ceil{(\frac{2\alpha - \varepsilon}{\varepsilon^2})(\frac{n \ln 2}{c} + \ln q)} \leq n$.

    We prove the lemma using the probabilistic method, like how we proved Lemma \ref*{sparselemma}. Let us construct each set $Q_i$ randomly by sampling $s$ uniformly random integers in $[n]$ without replacement. As in Lemma \ref*{sparselemma}'s proof, the sets $Q_1,\dots,Q_q$ witness the lemma's truth iff one cannot find a counterexample involving a maximal $B$. That is, the sets $Q_1,\dots,Q_q$ witness the lemma's truth iff no bad event $X_{B, C}$ occurs, where $X_{B, C}$ is the event that $B \subseteq [n]$ is a set of size $b = \floor{(\alpha - \varepsilon)n}$ and $C \subseteq [q]$ is a set of size $c$ such that $|Q_i \cap B| \geq \alpha s$ for all indices $i \in C$. Below, we show for every possible choice of $B$ and $C$ that the probability of the bad event $X_{B, C}$ is at most $2^{-n-c\log_2 q}$. This is enough as there are much fewer than $2^{n + c\log_2 q}$ ways to choose $B$ and $C$ ($\binom{n}{b} \leq \binom{n}{\ceil{n/3} - 1} = 2^{n - \Theta(n)}$ for $B$ and $\binom{q}{c} \leq q^c = 2^{c\log_2 q}$ for $C$), which means that with probability at least $1 - (2^{n - \Theta(n) + c\log_2 q})(2^{-n - c\log_2 q}) = 1 - 2^{-\Theta(n)}$ the bad event $X_{B, C}$ does not occur for any $B$ and $C$.
    
    Let $B \subseteq [n]$ be any set of size $b = \floor{(\alpha - \varepsilon)n}$, let $C \subseteq [q]$ be any set of size $c$, and let $Z = \sum_{i \in C}|Q_i \cap B|$. Suppose $b \geq 1$ without loss of generality, as the bad event $X_{B,C}$ is clearly impossible if $B$ is empty. Observe that $Z$ is the sum of $c$ independent random variables from the distribution $H(n, b, s)$ (the hypergeometric distribution with the parameters $n,b,s$), each distributed identically to the sum of some $\min(b, s)$ independent Bernoulli variables \cite{hp14}. So, $Z$ is distributed identically to the sum of some $c \cdot \min(b, s)$ independent Bernoulli variables, with $\mathbb{E}[Z] = \sum_{i \in C} \mathbb{E}[|Q_i \cap B|] = \sum_{i \in C}\frac{sb}{n} = \frac{csb}{n}$.

    Observe that if the bad event $X_{B, C}$ occurs, then $|Q_i \cap B| \geq \alpha s$ for all $i \in C$, which implies that $Z = \sum_{i \in C}|Q_i \cap B|\geq c\alpha s = (1 + \frac{\alpha n - b}{b})\mathbb{E}[Z]$. The Chernoff bound $\Pr[Z \geq (1 + \delta)\mathbb{E}[Z]] \leq \exp(\frac{-\mathbb{E}[Z] \delta^2}{2 + \delta})$ gives us $-\ln(\Pr[Z \geq c\alpha s]) = -\ln(\Pr[Z \geq (1 + \frac{\alpha n - b}{b})\mathbb{E}[Z]) \geq \frac{csb((\alpha n - b)/b)^2}{n(2 + (\alpha n - b)/b)} = \frac{csb((\alpha n - b)/b)^2}{n(\alpha n + b)/b} = \frac{cs(\alpha n - b)^2}{n(\alpha n + b)} \geq \frac{cs(\alpha n - (\alpha - \varepsilon)n)^2}{n(\alpha n + (\alpha - \varepsilon)n)} = \frac{cs\varepsilon^2}{2\alpha - \varepsilon} = \frac{c\ceil{(2\alpha - \varepsilon)/\varepsilon^2)((n\ln 2)/c + \ln q)}\varepsilon^2}{2\alpha - \varepsilon} \geq c(\frac{n \ln 2}{c} + \ln q) = n \ln 2 + c\ln q$. That is, we have $\Pr[Z \geq c\alpha s] \leq e^{-n\ln 2 - c\ln q} = 2^{-n - c\log_2 q}$. Finally, since the bad event $X_{B,C}$ implies $Z \geq c\alpha s$, we have $\Pr[X_{B, C}] \leq \Pr[Z \geq c\alpha s] \leq 2^{-n - c\log_2 q}$, as desired.
\end{proof}

\end{document}